\documentclass[review]{elsarticle}

\usepackage{lineno,hyperref}

\usepackage{amssymb}
\usepackage{amsmath}
\usepackage{amsthm}
\usepackage{algorithm}
\usepackage[noend]{algpseudocode}
\newtheorem{lemma}{Lemma}
\newtheorem{theorem}{Theorem}
\newtheorem{corollary}{Corollary}

\newtheorem{conjecture}{Conjecture}
\usepackage{float}

\modulolinenumbers[5]

\journal{Theoretial Computer Science}

\begin{document}

\begin{frontmatter}

\title{On Rivest-Vuillemin Conjecture for Fourteen Variables}


\author[mymainaddress]{Guang-mo Tong \corref{mycorrespondingauthor}}
\ead{gxt140030@utdallas.edu}
\author[mymainaddress]{Weili Wu}
\author[mymainaddress]{Ding-Zhu Du}

\cortext[mycorrespondingauthor]{Corresponding author}

\address[mymainaddress]{Department of Computer Science, University of Texas at Dallas, TX, USA}

\begin{abstract}
A boolean function $f(x_1,...,x_n)$ is \textit{weakly symmetric} if it is invariant under a transitive permutation group on its variables. A boolean function $f(x_1,...,x_n)$ is \textit{elusive} if we have to check all $x_1$,..., $x_n$ to determine the output of $f(x_1,...,x_n)$ in the worst-case. It is conjectured that every nontrivial monotone weakly symmetric boolean function is elusive, which has been open for a long time. In this paper, we show that this conjecture is true for $n=14$.
\end{abstract}


\end{frontmatter}

\linenumbers

\section{Introduction}
\label{sec:intro}
A boolean function $f(x_1,...,x_n)$ can be computed by a binary tree where each non-leaf node is labeled by a variable and the leaves are labeled by 0 or 1. For each non-leaf node, the two edges linking to its left and right child are labeled by 1 and 0, respectively. For any path from root to leave, every variable appears at most once. An input $\boldsymbol{x}$ of $f$ is a subset of $\{x_1,..., x_n\}$ where $x_i=1$ if and only if $x_i 
\in \boldsymbol{x}$. For each input $\boldsymbol{x}$ of $f$, its value can be computed according the decision tree of $f(x_1,...,x_n)$. That is, starting from the root, if the label of root is in $\boldsymbol{x}$, we go to its left child; otherwise we go to its right child. The above process is repeated until a leaf is reached and the function value of $x$ is given by the leaf's label. For each input $\boldsymbol{x}$, the computation time depends on the length of the corresponding root-leaf path, i.e., the number of checked variables. The depth of a decision tree is the maximum length among all root-leaf paths. One can see that for a certain boolean function $f$, there can be more than one decision trees. We denote by $D(f)$ the minimum depth among all its decision trees. A boolean function $f$ of $n$ variables is called elusive if $D(f)=n$. In other words, if $f$ is elusive, then for each of its decision trees, there exists an input $\boldsymbol{x}$ such that deciding $f(\boldsymbol{x})$ requires checking all the variables. A boolean function $f$ is monotone non-increasing if $f(\boldsymbol{x})=1$ implies $f(\boldsymbol{x^{'}})=1$ for each $x^{'} \subseteq \boldsymbol{x}$, and, similarly, $f$ is monotone non-decreasing if $f(\boldsymbol{x})=0$ implies $f(\boldsymbol{x^{'}})=0$ for each $x^{'} \subseteq \boldsymbol{x}$. For a permutation $\sigma$ on $\{1,...,n\}$ and an input $\boldsymbol{x}=\{x_{a_1},...,a_{a_m}\}$, let $\sigma(\boldsymbol{x})=\{x_{\sigma(a_1)},...,x_{\sigma(a_m)}\}$. A boolean function $f$ is $\sigma$-invariant if $f(\boldsymbol{x})=f(\sigma(\boldsymbol{x}))$ for every $\boldsymbol{x}$. For a group $G$ of permutations, $f$ is called $G$-invariant is $f$ is $\sigma$-invariant for every $\sigma \in G$. The symmetry of $f$ is characterized by its invariant group. An $G$-invariant boolean function $f$ is weakly symmetric if $G$ is transitive on $\{1,..., n\}$.

\textbf{Rivest-Vuillemin conjecture}: every nontrivial monotone weakly symmetric boolean function is elusive. 

In \cite{gao2002rivest}, \cite{gao1999rivest} and \cite{sui1999nontrivial}, it has been shown that it is also true when $n=6$, 10 and 12. Therefore, the Rivest-Vuillemin conjecture is true for $n$ less than 14. In this paper, we consider $n=14$. 

\section{Preliminaries}
\label{sec:pre}
Let $\overline{f}$ be the opposite function of $f$, i.e., $\overline{f}(\boldsymbol{x})= 0$ iff $f(\boldsymbol{x})=1$. It can be easily seen that $D(\overline{f})=D(f)$. Therefore, to prove the Rivest-Vuillemin conjecture, it suffices to consider monotone non-increasing boolean functions. Each monotone non-increasing boolean function $f(x_1,...,x_n)$ can be equivalently represented as an abstract simplicial complex $\varDelta_f$ on $n$ vertices defined by 
$\varDelta_f=\{ x | x \subseteq \{x_1,..., x_n\}~\text{and}~f(x)=1\}$. The faces in $\varDelta_f$ correspond to the true inputs of $f$. For an abstract simplicial complex $\varDelta$, the Euler characteristic $\chi(\varDelta)$ is defined as
\begin{equation}
\label{eq:euler}
\chi(\varDelta)=\sum_{i=1}^{n}(-1)^{i+1} r(\varDelta,i),
\end{equation}
where $r(G,i)=|\{x \in \varDelta~|~|x|=i\}|$.
Note that if $f$ is $G$-invariant then $G$ is a group of automorphisms on $\varDelta_f$. Kahn \textit{et al.} \cite{kahn1984topological} first observe that the evasiveness of a monotone boolean function $f$ is related to the topological property of $\varDelta_f$. 

\begin{theorem}{(\cite{kahn1984topological})}
\label{theorem:Kahn}
If a monotone  boolean function $f$ is not evasive, then $\varDelta_f$ is collapsible and therefore contractible and $Z_p$-acyclic.
\end{theorem}

For two primes $p$ and $q$, we denote by $\Psi_p^q$ the class of the finite group G with a normal subgroup $P \vartriangleleft H \vartriangleleft G$, such that $P$ is of $p$-power order, the quotient group $G/H$ is of $q$-power order, and the quotient group $H/P$ is cyclic; denote by $\Psi_p$ the class of the finite group G with a normal $p$-subgroup $P \vartriangleleft G$ such that the quotient group $G/P$ is cyclic. The following fixed-point theory is attributed to Oliver,

\begin{theorem}{(\cite{oliver1975fixed})}
\label{theorem:Oliver}
For a collapsible abstract complex $\varDelta$ with a group $G$ of automorphisms on $\varDelta$, if $G$ is cyclic or $G \in \Psi_p$ for some prime $p$, then $\chi(\varDelta^G)=1$; if $G \in \Psi_p^{q}$, then  $\chi(\varDelta^G)\equiv1~(mod~q)$, where
\begin{equation*}
\varDelta^G=\{\{H_1,...,H_k\}|H_1,...,H_k \text{~are orbits of~} G \text{~and~} H_1 \cup ... \cup H_k \in \varDelta\} \cup \{\emptyset\}.
\end{equation*}
\end{theorem}

We call the groups in  $\Psi_p^q$ and $\Psi_p$ as Oliver groups. The following result directly follows from Theorems \ref{theorem:Kahn} and \ref{theorem:Oliver}, 

\begin{theorem}
\label{theorem:topo}
For a monotone non-increasing $G$-invariant boolean function $f$, if $G$ is transitive, and,  $G$ is cyclic or $G$ is an Oliver group, then $f$ is elusive or trivial.
\end{theorem}

\begin{proof}
If $f$ is not evasive, then $\varDelta_f$ is collapsible and thus, by Theorem \ref{theorem:Oliver}, $\varDelta^G_f$ is non-empty. Since $G$ is transitive, that $\varDelta^G_f$ is non-empty implies that the only orbit of $G$ is in $\varDelta^G_f$, which means $\{x_1,...,x_n\} \in \varDelta_f$ and therefore $f$ is trivial.
\end{proof}

\begin{table}[h]
\caption{Minimal transitive permutation groups of degree 14 \label{table:list}}
{\begin{tabular}{ | p{1cm} | p{1.0cm} | p{7cm} |p{1cm}|}
\hline 
Group &index & Generators& Order\\ 
\hline 
$G_1$ &(1)& (1, 2, 3, 4, 5, 6, 7, 8, 9, 10, 11, 12, 13, 14)& 14\\ 
\hline
$G_2$ &(2)& (1, 3, 5, 7, 9, 11, 13)(2, 4, 6, 8, 10, 12, 14), \newline
			(1, 12)(2, 11)(3, 10)(4, 9)(5, 8)(6, 7)(13, 14)& 14 \\ 
\hline
$G_3$ &(6)& (3, 10)(5, 12)(6, 13)(7, 14), \newline
			(1, 3, 4, 7, 9, 11, 13)(2, 4, 6, 8, 10, 12, 14)& 56\\
\hline
$G_4$ &(12)&  (1,3,5,7,9,11,13), \newline 
				(1,2)(3,14,13,4)(5,12,11,6)(7,10,9,8),\newline
	(1, 6, 13, 8)(2, 9, 12, 5)(3, 4, 11, 10)(7, 14) & 169\\
\hline
$G_5$ &(30)&  (2, 4, 6, 8, 10, 12, 14), (2, 4, 8)(6, 12, 10), \newline
				(1, 8)(2, 5)(3, 4)(6, 13)(7, 14)(9, 12)(10, 11)& 1092\\
\hline
$G_6$ &(10)& (1, 5, 11, 10)(2, 9)(3, 8, 12, 4)(6, 14, 13, 7), \newline
				(1, 9, 5, 14)(2, 12, 7, 8)(3, 4, 10, 11)(6, 13)& 168\\

\hline
\end{tabular}}
{*The second column shows the index of each group in the GAP system.}
\end{table}
\section{Main result}
In this paper, we show the following result.

\begin{theorem}
\label{theorem:main}
Every nontrivial monotone non-increasing weakly symmetric boolean function of 14 variables is elusive. 
\end{theorem}

According to \cite{hulpke2005constructing}, there are totally 63 transitive groups of degree $14$ up to permutation isomorphism, where there are 6 minimal transitive groups shown in Table \ref{table:list}. These groups can be found in GAP system (\cite{gap2012gap}). Let $G_i$ , $1 \leq i \leq 6$, be the $i^{th}$ minimal transitive group. Therefore, any weakly symmetric boolean function with 14 variables must be invariant under at least one of the groups of $G_i$. Thus, to prove Theorem \ref{theorem:main}, it suffices to show that every nontrivial $G_i$-invariant monotone non-increasing boolean function is elusive. In the following, we will show that the first four groups are either cyclic or Oliver groups, which can be handled by Theorem \ref{theorem:main}, while the last two groups are neither cyclic nor Oliver groups for which we propose new techniques. In the rest of this section, $G_1$, $G_2$, $G_3$ and $G_4$ will be considered in Sec. \ref{subsec:g_1-4}, and, $G_5$ and $G_6$ will be discussed in Sec. \ref{subsec:g_5-6}, respectively.

\subsection{$G_1$, $G_2$, $G_3$ and $G_4$ }
\label{subsec:g_1-4}
\begin{lemma}
Every non-trivial monotone non-increasing $G_1$-invariant boolean function $f$ is elusive.
\end{lemma}
\begin{proof}
Since $G_1$ is cyclic, the lemma directly follows from Theorem \ref{theorem:topo}.
\end{proof}

\begin{lemma}
Every non-trivial monotone non-increasing $G_2$-invariant boolean function $f$ is elusive.
\end{lemma}
\begin{proof}
Let 
\begin{equation*}
a=(1, 3, 5, 7, 9, 11, 13)(2, 4, 6, 8, 10, 12, 14)
\end{equation*}
and
\begin{equation*}
b=(1, 12)(2, 11)(3, 10)(4, 9)(5, 8)(6, 7)(13, 14).
\end{equation*} 
$G_2=<a,b>$.
Let $P=<a>$ be the subgroup of $G_2$ generated by $a$. Since $bab=a^6$, $P$ has an index of $2$. Therefore, $P$ is a normal 7-subgroup and $G/P$ is a cyclic group. Thus, $G_2 \in \Psi_7$. By Theorem \ref{theorem:main}, every $G_2$-invariant monotone boolean function $f$  is elusive.
\end{proof}

\begin{lemma}
Every non-trivial monotone non-increasing $G_3$-invariant boolean function $f$ is elusive.
\end{lemma}
Let 
\[a=(2,9)(4,11)(5,12)(6,13),~b=(1,8)(2,9)(5,12)(7,14)\] and\[c=(3,10)(5,12)(6,13)(7,14).\] One can check that group $<a, b, c>$ is a normal 2-subgroup of $G_3$ of index $7$. Therefore $G_3 \in \Psi_2$ and by Theorem \ref{theorem:main}, every nontrivial $G_3$-invariant monotone boolean function is elusive.

\begin{lemma}
Every non-trivial $G_4$-invariant monotone non-increasing boolean function $f$ is elusive.
\end{lemma}
Let \[a=(3,13)(4,14)(5,11)(6,12)(7,9)(8,10),~b=(1,3,5,7,9,11,13)\] and \[c= (2,4,6,8,10,12,14).\]
Group $H=<a,c,b>$ is a normal subgroup of $G_4$. Because $|H|=98$ and $|G_5|=196$, $G/H$ is of order $2$. Let \[d=(1,3,5,7,9,11,13)\] and \[e=
(2,4,6,8,10,12,14).\] Group $P=<d,e>$ is a normal subgroup of $H$. Because $|H|=98$ and $|P|=49$, $P$ is of 7-power and $H/P$ is cyclic. Therefore, $G_5 \in \Psi_7^2$, and thus, by Theorem \ref{theorem:main}, proved.

\subsection{$G_5$ and $G_6$}
\label{subsec:g_5-6}
In this section we consider $G_5$ and $G_6$. Note that $G_5$ and $G_6$ are not cyclic and furthermore they are not solvable. Thus, the existing technique can not be applied to prove the evasiveness of an $G_5$ or $G_6$-invariant monotone boolean function. In the following, we proceed in another approach.

The following result is well-known and intuitive.
\begin{lemma}
\label{lemma:f_x=1}
For a $G$-invariant boolean function $f(x_1,...,x_n)$ where $G$ is transitive, if, for some $x_a \in \{x_1,...,x_n\}$, $f_{x_a=1}$ is elusive, then $f$ is elusive. 
\end{lemma}
\begin{proof}
Let $T$ be an arbitrary decision tree of $f$ and $\sigma$ a permutation in $G$. By relabeling the variable $x_i$ in $T$ by $x_{\sigma(i)}$, we obtain another decision tree denoted by $\sigma(T)$. Since the label of leaves remain unchanged, $T$ and $\sigma(T)$ have the same length. Suppose the root of $T$ is $x_b$. Because $G$ is transitive, there is a permutation $\overline{\sigma}$ in $G$ such that $\overline{\sigma}(b)=a$. Therefore, $\overline{\sigma}(T)$ is a decision tree of $f$ with root $x_a$. Let $\overline{\sigma}(T)_a^1$ be the left subtree of $\overline{\sigma}(T)$ on the root $x_a$. Because $f_{x_a=1}$ is elusive and $\overline{\sigma}(T)_a^1$ is a decision tree of $f_{x_a=1}$, the depth of $\overline{\sigma}(T)_a^1$ is $n-1$, which implies $\overline{\sigma}(T)$ is of depth $n$ and so is $T$. Since $T$ is arbitrarily selected, every the decision tree of $f$ has a depth of $n$. 
\end{proof}

One can check that if $f_{x_a=1}$ is elusive for some $x_a$, then $f_{x_i=1}$ is elusive for every $x_i$. Although it is not easy to directly prove the elusiveness of an $G_5$-invariant function $f$, we are able to show that $f_{x_a=1}$ is elusive for some $x_a$.  
\begin{lemma}
\label{lemma:new}
If a nontrivial monotone Boolean function $f$ of $p+1$ variables is invariant under a transitive group $G$ where $p$ is a prime and $|G|=p*k$ where $k$ is an integer and $p$ does not divide $k$, then $f$ is elusive.
\end{lemma}
\begin{proof}
By the Sylow p-subgroup theory, $G$ has cyclic subgroup $H$ such that $|H|=p$ and $H$ is the stabilizer of some $x_a$. Therefore, $f_{x_a=1}$ is invariant under $H$. Since $H$ is cyclic and transitive on $\{x_1,...,x_n\}\setminus\{x_a\}$, $f_{x_a=1}$ is elusive according to Theorem \ref{theorem:main}. Combining Lemma \ref{lemma:f_x=1}, $f$ is elusive.
\end{proof}

Since $|G_5|=13*3*7*2^2$, the following directly follows. 
\begin{corollary}
Every non-trivial $G_5$-invariant monotone boolean function is elusive.
\end{corollary}

Now let us consider $G_6$. 

First we consider the subgroups of $G_6$. Note that if $f$ is invariant under $G_6$ then it is also invariant under any subgroup of $G_6$. Therefore, if $f$ is not elusive, then by Theorem \ref{theorem:Kahn} $\varDelta_f$ is collapsible and thus for any subgroup $G^{'} < G$, if $G^{'}$ is cyclic or $G^{'} \in \Psi_p$, $\chi(\varDelta^{G^{'}})=1$; if $G^{'} \in \Psi_p^q$ , $\chi(\varDelta^{G^{'}})\equiv 1 ~(mod~q)$. Now let us consider the 11 subgroups of $G_6$ listed in Table \ref{table:g4_sublist_1}, \ref{table:g4_sublist_2} and \ref{table:g4_sublist_3} in \ref{sec:sub_g_6}, denoted by $G_6^1,...,G_6^{11}$. By the above analysis, if an $G_6$-invariant monotone boolean function $f$ is not elusive, then $\chi(\varDelta^{G_6^i}_{f})=1$ for $1 \leq i \leq 10$ and $\chi(\varDelta^{G_6^{11}}_{f})\equiv 1 ~(mod~2)$. Note that when $G^{'}$ is identity, $\chi(\varDelta_{f}^{G^{'}})=1$ implies $\chi(\varDelta_{f})=1$.

Second, we consider $f$ restricted on one if its variables. For every $x_a \in\{x_1,...,x_n\}$, let  $\text{Link}(\varDelta,x_a)$ and $\text{Deletion}(\varDelta,x_a)$ be the subcomplexes of $\varDelta$, which are defined as 
\begin{equation*}
\text{Link}(\varDelta,x_a)=\{x - \{x_a\}|~x_a \in x, x \in \varDelta \},
\end{equation*}
and
\begin{equation*}
\text{Deletion}(\varDelta,x_a)=\{x |~x_a \notin x, x \in \varDelta \}.
\end{equation*}
It can be easily checked that $\text{Link}(\varDelta_f,x_a)=\varDelta_{f_{x_a=1}}$. Thus, if $f$ is not elusive, then $f_{x_a=1}$ is not elusive and therefore $\varDelta_{f_{x_a=1}}$ is collapsible, which implies $\chi(\varDelta_{f_{x_a=1}})=1$. Due to the weakly symmetry, once $r(\varDelta_{f},k)$ is known to us, the following relationship allows us to compute $r(\varDelta_{f_{x_a=1}},k)$ efficiently,
\begin{equation*}
n \cdot r(\varDelta_{f_{x_a=1}},k)=k \cdot r(\varDelta_{f},k).
\end{equation*}

By the above analysis, if $f$ is $G_6$-invariant but not elusive, the followings must be satisfied:
\begin{itemize}
\item $\chi(\varDelta^{G_6^i}_{f})=1$ for $1 \leq i \leq 10$ and $\chi(\varDelta^{G_6^{11}}_{f})\equiv 1 ~(mod~2)$;
\item $\chi(\varDelta_{f_{x_1=1}})=1$.
\end{itemize}
Our goal is to verify that such an $f$ does not exist, i.e.,  

\begin{theorem}
\label{theorem:G_6}
There is no monotone non-increasing $G_6$-invariant boolean function $f$ such that $\chi(\varDelta^{G_6^i}_{f})=1$ for $1 \leq i \leq 10$, $\chi(\varDelta^{G_6^{11}}_{f})\equiv 1 ~(mod~2)$, and
$\chi(\varDelta_{f_{x_1=1}})=1$.
\end{theorem}

To this end, let us consider the pattern of the orbits generated by $G_6$. We call a $k$-subsets of $\{x_1,...,x_{14}\}$ as a $k$-tuple. Let $T_k$ be the set of all k-tuples. The orbits on the k-tuples generated by $G_6$ are called $k$-orbits. A $k$-orbit is a subset of $T_k$. For example, $G_6$ forms two 2-orbits on the 2-tuples where one orbit has 84 elements and another one has 7 elements.  For a $k_1$-orbit $O_{1}$ and a $k_2$-orbit $O_{2}$ where $k_1<k_2$, if there exists two tuples $t_1$ and $t_2$ such that $t_1 \in O_1$, $t_2 \in O_2$ and $t_1 \subseteq t_2$, we say $O_1$ is smaller than $O_2$ or equivalently $O_2$ is larger than $O_1$, denoted by $O_1 \leq O_2$. Let $\text{Upper}(O)$ be the set of orbits which are larger than orbit $O$, i.e, $\text{Upper}(O)=\{O^{'}| O \leq O^{'}\}$, and similarly, let $\text{Lower}(O)=\{O^{'}| O^{'} \leq O \}$. Because $f$ is invariant under $G_6$, the tuples in the same k-orbit must have the same function value. We say a k-orbit is a T-orbit (resp. F-orbit) if the tuples in it result true (resp. false) function value. Due to the monotonicity, if an orbit $O$ is a T-orbit, then the orbits in $\text{Lower}(O)$ must be T-orbits; if an orbit $O$ is an F-orbit, then the orbits in $\text{Upper}(O)$ must be F-orbits. The relationship between the orbits under $G_6$ are shown in Fig. \ref{fig:pattern}, where if there is an edge between two orbits, then one is larger than the other. Since $G_6$ is given explicitly, the relationship of orbits can be easily computed by program. We number the orbits consistently and let $O_{i,j}$ be the $j$-th $i$-orbit, $1 \leq i \leq 14$ and $j \geq 0$. 
\begin{figure}[h]
\begin{center}
\includegraphics[width=3in]{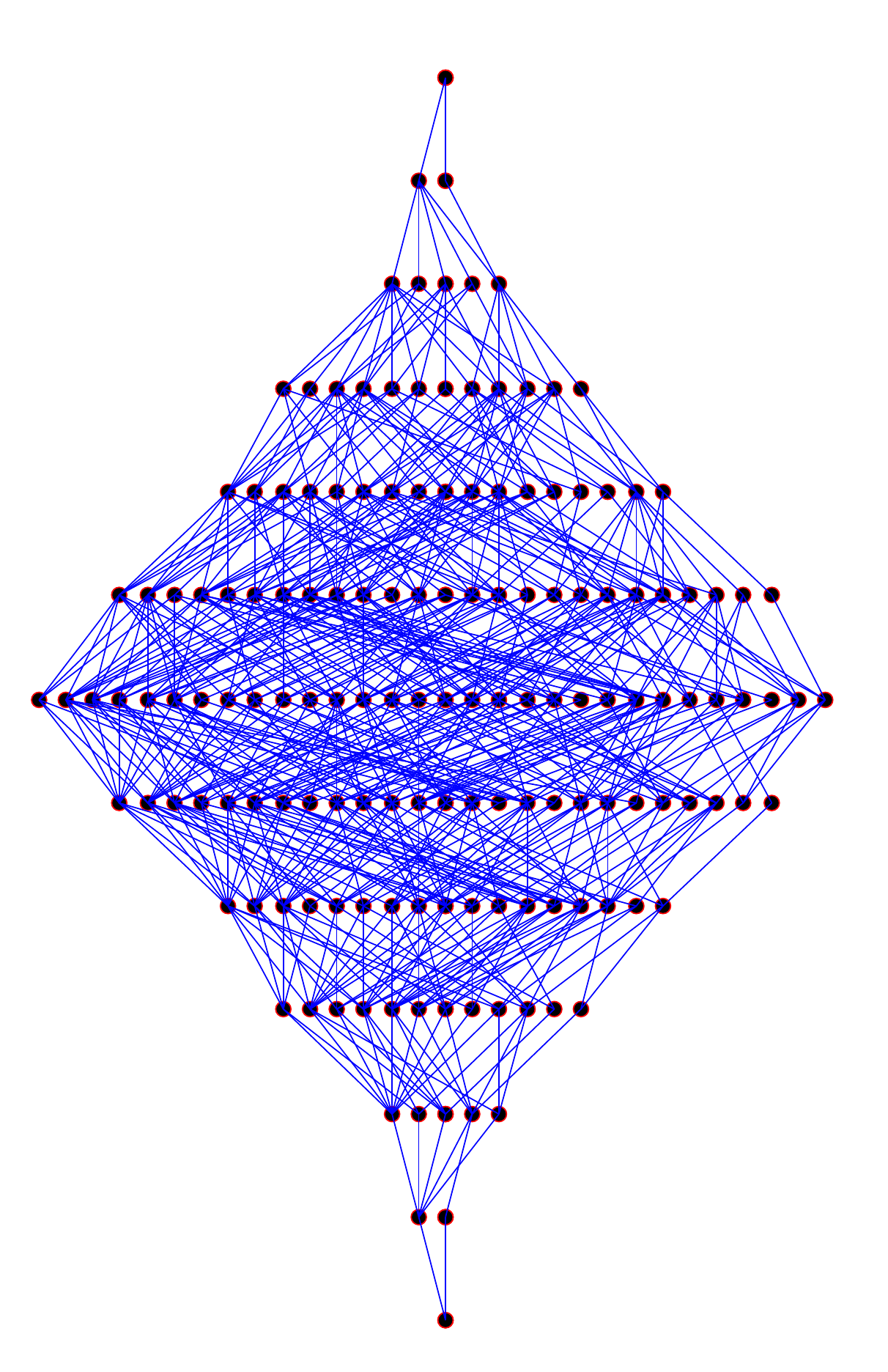} 
\end{center} 
\caption{\textbf{Orbits pattern of $G_6$}}
\label{fig:pattern}
\end{figure}

As shown in Fig. \ref{fig:pattern}, there are totally 158 orbits under $G_6$, which means there are $2^{158}$ boolean functions invariant under $G_6$. Thus, it is impracticable to directly check all those functions. In the following, we will show that it suffices to consider a small number of functions, due to a case by case analysis. Initially, all the orbits are called free orbits. In order to satisfy the conditions in Theorem \ref{theorem:G_6}, some orbits have to be determined as T-orbits or F-orbits. For example, because $\chi(\varDelta^{G_6^{11}}_{f})\equiv 1 ~(mod~2)$ and $G_6^{11}$ forms two orbits $H_{11}^1$ and $H_{11}^2$ on 1-tuples, $\varDelta^{G_6^{11}}_{f}$ is either $\{\{H_{11}^1\},\emptyset\}$, $\{\{H_{11}^2\},\emptyset\}$ or $\{\{H_{11}^1\},\{H_{11}^2\},\{H_{11}^1,H_{11}^2\}, \emptyset\}$. However, $f$ is nontrivial which means $\{H_{11}^1,H_{11}^2\} \notin \varDelta_f$. Therefore only $\{\{H_{11}^1\},\emptyset\}$ and $\{\{H_{11}^2\},\emptyset\}$ are possible. If $\{\{H_{11}^1\},\emptyset\}$ is the case, then according to Table \ref{table:g4_sublist_3}, the 6-tuple $H_{11}^1=\{4,5,11,7,14,12\}$ must be a true input and $H_{11}^2=\{1,9,3,10,6,8,2,12\}$ must be false input. Because $H_{11}^1$ and $H_{11}^2$ belongs to orbits $O_{6.24}$ and $O_{8.24}$, respectively, orbits in $\text{Lower}(O_{6.24})$\footnote{The index of the orbits does not matter as long as long it is consistent. Here we use the index generated by our program for illustration.} should be T-orbits and the orbits in $\text{Upper}(O_{8.24})$ should be F-tuples. Therefore, by checking the conditions in Theorem \ref{theorem:G_6} we can keep determining the type of the orbits. After checking $\chi(\varDelta^{G_6^{1}}_{f})$, there will be no free orbits, which implies the function is completely determined. Finally, we can check $\chi(\varDelta_{f_{x_1=1}})$ to see whether there is an $G_6$-invariant function satisfying all the conditions in Theroem \ref{theorem:G_6}. Specifically, we will first check the $\chi(\varDelta^{G_6^{i}}_{f})$ where $G_6^{i}$ has the fewest orbits. The checking framework is shown in Algorithm \ref{alg:check}. The whole process is done by a Java-based programming combing with the GAP system. It turns out that no type setting of the orbits can satisfy all the conditions. In the appendix, we provide an example to show one branch of the computing. 

\begin{algorithm}[h]
\caption{ \textbf{check(T-orbits, F-orbits, index)}}
\label{alg:check}
\begin{algorithmic}[1]
\If {(index $\leq$ 11) }
		\For {each feasible case such that $\chi(\varDelta^{G_6^{index}}_{f})=1$}  
		\State	update T-orbits, F-orbits;
		\State	check(T-orbits, F-orbits, index+1);
		\EndFor
\Else	\If {(index $==$ 12)}
		\State compute $\chi(\varDelta^{G_6}_{f_{x_1=1}})$ according to $T$-orbits, $F$-orbits.
			\If { $\chi(\varDelta^{G_6}_{f_{x_1=1}}) ==1 $}
			\State \Return a feasible boolean function found;
			\EndIf
		\EndIf
\EndIf
\end{algorithmic}
%\vspace{-1mm}
\end{algorithm}

\section{Discussion}
There has been other works that manage to verify the elusiveness of a boolean function by programming. For example, in \cite{lutz2001some}, the authors have checked the evasiveness of a $G$-variant boolean function for some $G$ by enumerating the complexes and checking the $Z_p$-acyclic. However, given a group $G$, checking all the $G$-invariant boolean functions in brute force is extremely time consuming and the method proposed in \cite{lutz2001some} cannot deal with the case for $14 \leq n$. The checking framework proposed in this paper in more efficient and fundamentally it reveals how the weakly symmetry forces the complex to be a simplex.

The initial conjecture made by Rivest and Vuillemin \cite{rivest1975generalization} is that every weakly symmetric boolean function $f$ with $f(\emptyset)\neq f(\{x_1,...,x_n\})$ is elusive, which is negated by Illies by a counterexample \cite{kahn1984topological}. Aigner \cite{aigner1988combinatorial} further modify the conjecture into its current version by adding the condition of monotonicity. Due to the monotonicity, a boolean function $f$ is equivalent to an abstract simplicial complex $\varDelta_f$. The critical observation by Kahn et al. \cite{kahn1984topological} shows if $f$ is non-elusive then $\varDelta_f$ must be collapsible and therefore contractible, which enables us to apply the fixed-point theory. For a contractible abstract simplicial complex with a automorphism group $G$, Oliver \cite{oliver1975fixed} shows that under certain circumstance (i.e., Oliver group) there exists a face which is fixed by $G$. Therefore, if the invariant group is an Oliver transitive group, $\varDelta_f$ must be a simplex, which means $f$ is trivial. When $G$ is not a Oliver group, we may apply the fixed-point theory to its subgroups, as shown in this paper. Given the invariant group, we have although large but limited number of boolean functions. While applying the fixed-point theory to the subgroups, we are able to eliminate the complexes that are not collapsible. Kahn et al. \cite{kahn1984topological} propose a conjecture that a non-empty collapsible weakly symmetric complex must be a simplex. The truth of this conjecture yields the truth of Revest-Vuillemin conjecture\footnote{As mentioned in \cite{kahn1984topological}, Oliver has provided a plausibility argement for the falsity of this conjecture, in personal communication.}. 

Finally, we remark a stronger condition. Note that the Link and Deletion of a non-evasive weakly symmetric complex must be non-evasive. Thus, the following conjecture implies the Revest-Vuillemin conjecture:
\begin{conjecture}
For a non-empty weakly symmetric complex $\varDelta$, if $Link(\varDelta,x)$ and $Deletion(\varDelta,x)$ are all collapsible, then $\varDelta$ is a simplex.
\end{conjecture}

The condition in the above statement is stronger and it has a clear meaning that the complex is not only collapsible but also be able to collapse to a point along a certain sequence of collapses.

\appendix
\section{Subgroups of $G_6$}
\label{sec:sub_g_6}
\begin{table}[H]
\caption{Subgroups of $G_6$ \label{table:g4_sublist_1}}
\small
{\begin{tabular}{ | p{0.8cm} | p{1cm} | p{7.5cm} |p{1cm}|}
\hline 
Group &index & Generators and orbits & type\\ 
\hline 
$G_6^1$ &(1)& identity& \\ 
\hline
$G_6^{2}$ &(2)& \textbf{Generators}: \newline
			(2,4)(3,10)(5,6)(7,14)(9,11)(12,13) \newline
			\textbf{orbits}: \newline
			$H_{2}^1$: $\{ 1 \}\in O_{1.0}$;~~~~~$H_{2}^2$: $\{2, 4\}\in O_{2.0}$; \newline
			$H_{2}^3$: $\{ 3,10 \}\in O_{2.1}$; $H_{2}^4$: $\{5, 6\}\in O_{2.0}$; \newline
			$H_{2}^5$: $\{7, 14 \}\in O_{2.1}$; $H_{2}^6$: $\{8\}\in O_{1.0}$; \newline
			$H_{2}^7$: $\{9, 11 \}\in O_{2.0}$; $H_{2}^8$: $\{12, 13\}\in O_{2.0}$ 
			& cyclic
			\\ 
\hline

$G_6^{3}$ &(23)&   \textbf{Generators}: \newline
			(2,3,6)(4,7,12)(5,11,14)(9,10,13)\newline 
			\textbf{orbits}: \newline
			$H_{3}^1$: $\{1 \} \in O_{1.0}$; ~~~~~~$H_{3}^2$: $\{2, 3, 6 \} \in O_{3.2}$; \newline
			$H_{3}^3$:$\{ 4, 7, 12 \} \in O_{3.1}$; $H_{3}^4$: $\{5, 11, 14\} \in O_{3.4}$; \newline 
			$H_{3}^5$: $\{8 \} \in O_{1.0}$; ~~~~~~$H_{3}^6$: $\{9, 10, 13\} O_{3.2}$ & cyclic
			\\ 
\hline

$G_6^{4}$ &(51)& \textbf{Generators}: \newline
			(1,8)(2,13)(3,10)(4,12)(5,11)(6,9), \newline
			(1,8)(2,12)(4,13)(5,9)(6,11)(7,14) \newline
			\textbf{orbits}:\newline
			$H_{4}^1$: $\{ 1, 8 \} \in O_{2.1}$,  	$H_{4}^2$: $\{ 2, 13, 12, 4 \} \in O_{4.10}$,  \newline
			$H_{4}^3$: $\{ 3, 10 \} \in O_{2.1}$,  $H_{4}^4$: $\{ 5, 11, 9, 6 \} \in O_{4.0}$,  \newline
			$H_{4}^5$: $\{ 7, 14 \} \in O_{2.1}$& $\Psi_2$
			\\
\hline
$G_6^{5}$ &(58)& \textbf{Generators}: \newline  
				(2,11)(3,7)(4,9)(5,12)(6,13)(10,14), \newline
				 (2,4)(3,10)(5,6)(7,14)(9,11)(12,13) \newline
				 \textbf{orbits}: \newline
				$H_{5}^1$: $\{ 1 \} \in O_{1.0}$, $H_{5}^2$: $\{ 2, 4, 11, 9 \} \in O_{(4.11)}$, \newline
				$H_{5}^3$: $\{ 3, 10, 7, 14 \}\in O_{4.11}$, \newline
				$H_{5}^4$: $\{ 5, 6, 12, 13 \} \in O_{4.11}$, $H_{5}^5$: $\{ 8 \} \in O_{1.0}$&$\Psi_2$
				\\

\hline
\end{tabular}}
\end{table}
\begin{table}[H]
\caption{Subgroups of $G_6$ (cont.)\label{table:g4_sublist_2}}
\small
{\begin{tabular}{ | p{1cm} | p{1cm} | p{7cm} |p{1cm}|}
\hline 
Group & index & Generators and orbits on 1-tuples& type\\ 
\hline
$G_6^{6}$ & (65)&  \textbf{Generators}: \newline  
	(1,8)(2,5,4,6)(3,7,10,14)(9,12,11,13) \newline 
	\textbf{orbits}: 
	$H_{6}^1$: $\{ 1, 8 \} \in O_{2.1}$; $H_{6}^2$: $\{ 2, 5, 4, 6 \} \in O_{4.7}$, \newline
	$H_{6}^3$: $\{ 3, 7, 10, 14 \} \in O_{4.11}$,\newline
	$H_{6}^4$: $\{ 9, 12, 11, 13 \} \in O_{4.7}$ & cyclic
	\\	
\hline 
$G_6^{7}$ & (86)&  \textbf{Generators}: \newline  
				(2,3,6)(4,7,12)(5,11,14)(9,10,13) \newline
				(1,8)(2,13)(3,10)(4,12)(5,11)(6,9) \newline
				\textbf{orbits}: \newline 
				 $H_{7}^1$: $\{ 1, 8 \} \in O_{2.1}$, \newline
				$H_{7}^2$: $\{ 2, 3, 13, 6, 10, 9 \} \in O_{6.23}$, \newline
				$H_{7}^3$: $\{ 4, 7, 12 \} \in O_{3.1}$, $H_{7}^4$: $\{ 5, 11, 14 \} \in O_{3.4}$	& $\Psi_2$
				\\
	
\hline
$G_6^{8}$ & (142)&   \textbf{Generators}: \newline  
				(1,14)(2,9)(4,6)(5,12)(7,8)(11,13) \newline
				 (1,8)(2,10)(3,9)(4,7)(6,13)(11,14) \newline
				 \textbf{orbits}: \newline  
				$H_{8}^1$: $\{ 1, 14, 8, 11, 7, 13, 4, 6 \} \in O_{8,24}$, \newline
				$H_{8}^2$: $\{ 2, 9, 10, 3 \} \in O_{4.11}$, $H_{8}^3$: $\{ 5, 12\}\in O_{2.1}$	& $\Psi_2$
	\\
\hline
$G_6^{9}$ & (149)&  \textbf{Generators}: \newline  
				(1,11,3)(2,6,14)(4,10,8)(7,9,13)\newline
				(1,3)(2,9)(4,5)(6,13)(8,10)(11,12)\newline
				\textbf{orbits}: \newline  
				$H_{9}^1$: $\{ 1, 11, 3, 12 \} \in O_{4.10}$, \newline
				$H_{9}^2$: $\{ 2, 6, 9, 14, 13, 7 \} \in O_{6.24}$, \newline
				$H_{9}^3$: $\{ 4, 10, 5, 8 \} \in O_{4.10}$		& $\Psi_2$
	\\

\hline
\end{tabular}}
\end{table}

\begin{table}[H]
\caption{Subgroups of $G_6$ (cont.) \label{table:g4_sublist_3}}
\small
{\begin{tabular}{ | p{0.8cm} | p{1cm} | p{7.5cm} |p{1cm}|}
\hline 
Group &index & Generators and orbits on 1-tuples& type\\ 
\hline
$G_6^{10}$ & (157)&  \textbf{Generators}: \newline  
				 (1,6,4)(2,12,3)(5,10,9)(8,13,11),\newline
				 (1,12,6)(3,7,4)(5,13,8)(10,14,11) \newline
				 \textbf{orbits}: \newline  
				 $H_{10}^1$: $\{ 1, 6, 12, 4, 3, 2, 7 \} \in O_{7.20}$, \newline
				 $H_{10}^2$: $\{ 5, 10, 13, 9, 14, 11, 8  \} \in O_{7.27}$		& $\Psi_7$
	\\
\hline
$G_6^{11}$ & (165)&    \textbf{Generators}: \newline  
					(1,9,10)(2,3,8)(4,5,7)(11,12,14),\newline
				  (1,3,9,6)(2,13,8,10)(4,11)(5,14,12,7) \newline
				  \textbf{orbits}: \newline  
				$H_{11}^1$: $\{ 4, 5, 11, 7, 14, 12 \} \in O_{6.24},$ \newline
				$H_{11}^2$:$ \{ 1, 9, 3, 10, 6, 8, 2, 13 \} \in O_{8.24}$
			& $\Psi_2^2$
	\\	
\hline
\end{tabular}}
\end{table}

\section{A case study for Algorithm \ref{alg:check}}
\textbf{Step 1:} As discussed in Sec. \ref{subsec:g_5-6}, in order to meet that $\chi(\varDelta^{G_6^{11}}_{f})\equiv 1 ~(mod~2)$, there are two cases to consider. Suppose  $\chi(\varDelta^{G_6^{11}}_{f})=\{\{H_{11}^1\},\emptyset\}$ is selected. Let $\Theta_T$ and $\Theta_F$ be the set of the T-orbits and F-orbits that be currently determined
%, and let $\Theta_E^1$ be the set of orbits that have not been set a type
. Thus,  $\Theta_T=\text{Lower}(O_{6.24})$ and $\Theta_F=\text{Upper}(O_{8.24})$. According to the relationship in Fig. \ref{fig:pattern}, currently,
\begin{eqnarray*}
\Theta_T=&\{&O_{1.0},O_{2.0},O_{2.1},O_{3.1},O_{1.3},O_{1.4},O_{4.9},O_{4.11},O_{5.16,}O_{6.24}\}. \\
\Theta_F=&\{&O_{8.24}, O_{9.16}, O_{10.9}, O_{10.11}, O_{11.1}, O_{11.3}, O_{11.4}, O_{12.0}, O_{12.1,}, O_{13.0}, O_{14.0}\}.
\end{eqnarray*}
%, and $\Theta_E^1=\cup_{i,j}O_{i,j} \setminus (\text{Lower}(O_{6.24}) \cup \text{Upper}(O_{8.24}))$

\textbf{Step 2:} Now we consider $\varDelta^{G_6^{10}}_{f}$. Note that $H_{10}^1 \in O_{7.20}$ and $H_{10}^2 \in O_{7.27}$. Because neither of $O_{7,20}$ or $O_{7,27}$ is in $\Theta_T$ or $\Theta_F$, we have two cases to consider. One is $\chi(\varDelta^{G_6^{10}}_{f})=\{\{H_{10}^1\},\emptyset\}$ and the other is $\chi(\varDelta^{G_6^{10}}_{f})=\{\{H_{10}^2\},\emptyset\}$. Suppose  $\chi(\varDelta^{G_6^{10}}_{f})=\{\{H_{10}^1\},\emptyset\}$ is true. Now more orbits can be determined as T- or F-orbits. In particular, $\Theta_T=\Theta_T \cup \text{Lower}(O_{7.20})$, $\Theta_F=\Theta_F \cup \text{Upper}(O_{7.27})$. According to the relationship in Fig. \ref{fig:pattern},
\begin{eqnarray*}
\Theta_T=&\{&O_{1.0}, O_{2.0}, O_{2.1}, O_{3.0}, O_{3.1}, O_{3.2}, O_{3.3}, O_{3.4}, O_{4.0}, O_{4.2}, O_{4.3}, O_{4.9}， \\
&& O_{4.11}, O_{5.1},  O_{5.16}, O_{6.7}, O_{6.24}, O_{7.20} \}. \\
\Theta_F=&\{&O_{7.27}, O_{8.12}, O_{8.24}, O_{9.3}, O_{9.16}, O_{10.0}, O_{10.3}, O_{10.5}, O_{10.9}, O_{10.11}, \\
&&O_{11.0}, O_{11.1}, O_{11.2},  O_{11.3}, O_{11.4}, O_{12.0}, O_{12.1,}, O_{13.0}, O_{14.0}\}.
\end{eqnarray*}

\textbf{Step 3:} Now we consider $\varDelta^{G_6^{9}}_{f}$. One can check that   $H_9^1 \in O_{4.10}$, $H_9^2 \in O_{6.24} \in \Theta_T$, $H_9^1 \cup H_9^3 \in O_{8.24} \in \Theta_F$, $H_9^2 \cup H_9^3 \in O_{8.24} \in \Theta_F$, and $H_9^1 \cup H_9^2 \in O_{10.6}$. Therefore, in order to make $\chi(\varDelta^{G_6^{9}}_{f})=1$ there are only two possible cases, $\varDelta^{G_6^{9}}_{f}=\{\{H_{9}^2\},\emptyset\}$ and $\varDelta^{G_6^{9}}_{f}=\{\{H_{9}^1\},\{H_{9}^2\},\{H_{9}^3\},\{H_{9}^1,H_{9}^2\},\{H_{9}^3,H_{9}^2\},\emptyset\}$. Suppose  $\varDelta^{G_6^{9}}_{f}=\{\{H_{9}^1\},\emptyset\}$ is true. Then there one new F-orbits and no T-orbits added . Thus, $\Theta_F=\Theta_F \cup \text{Upper}(O_{4.10})$. According to the relationship in Fig. \ref{fig:pattern},
\begin{eqnarray*}
\Theta_T=&\{&O_{1.0}, O_{2.0}, O_{2.1}, O_{3.0}, O_{3.1}, O_{3.2}, O_{3.3}, O_{3.4}, O_{4.0}, O_{4.2}, O_{4.3}, O_{4.9}， \\
&& O_{4.11}, O_{5.1},  O_{5.16}, O_{6.7}, O_{6.24}, O_{7.20} \}. \\
\Theta_F=&\{& O_{4.10}, O_{5.5}, O_{5.10}, O_{6.1}, O_{6.8}, O_{6.13}, O_{6.16}, O_{6.20}, \\
&& O_{7.0}, O_{7.1}, O_{7.7}~O_{7.11},  O_{7.13}, O_{7.19}, O_{7.21},  O_{7.22},  O_{7.26}, O_{7.27}, O_{7.29}, \\
&& O_{8.0}\sim O_{8.7}, O_{8.11}\sim O_{8.16}, O_{8.19}\sim O_{8.22}, O_{8.24}, \\
&& O_{9.0}\sim O_{9.11}, O_{9.13}\sim O_{9.16},\\
&& O_{10.0}\sim O_{10.11}, O_{11.0}\sim O_{11.4}, O_{12.0}, O_{12.1,}, O_{13.0}, O_{14.0}\}.
\end{eqnarray*}

\textbf{Step 4:} Now we consider $\varDelta^{G_6^{7}}_{f}$. One can check that $H_7^1 \in O_{2.1} \in \Theta_T$,  $H_7^2 \in O_{6.23}$,  $H_7^3 \in O_{3.1} \in \Theta_T$,  $H_7^4 \in O_{3.4} \in \Theta_T$, $H_7^1 \cup H_7^3 \in O_{5.3}$, $H_7^1 \cup H_7^4 \in O_{5.15}$, $H_7^3 \cup H_7^4 \in O_{6.24} \in \Theta_T$, and for all $10 \leq k$, $O_{k,j} \subseteq \Theta_F$. Therefore, in order to make $\chi(\varDelta^{G_6^{7}}_{f})$ be 1, there are four possible cases,
\begin{enumerate}
\item $\varDelta^{G_6^{7}}_{f}=\{\{H_{7}^1\},\{H_{7}^3\},\{H_{7}^4\},\{H_{7}^3,H_{7}^4\},\{H_{7}^1,H_{7}^3\},\emptyset\}$;
\item $\varDelta^{G_6^{7}}_{f}=\{\{H_{7}^1\},\{H_{7}^3\},\{H_{7}^4\},\{H_{7}^3,H_{7}^4\},\{H_{7}^1,H_{7}^4\},\emptyset\}$;
\item $\varDelta^{G_6^{7}}_{f}=\{\{H_{7}^1\},\{H_{7}^3\},\{H_{7}^4\},\{H_{7}^3,H_{7}^4\},\{H_{7}^1,H_{7}^3\},\{H_{7}^1,H_{7}^4\},\{H_{7}^1,H_{7}^3,H_{7}^4\},\emptyset\}$;
\item $\varDelta^{G_6^{7}}_{f}=\{\{H_{7}^1\},\{H_{7}^2\},\{H_{7}^3\},\{H_{7}^4\},\{H_{7}^3,H_{7}^4\},\{H_{7}^1,H_{7}^3\},\{H_{7}^1,H_{7}^4\},\emptyset\}$.
\end{enumerate}
Suppose the first one is true. Then there is one new T-orbits and two new F-orbits. Thus, $\Theta_T=\Theta_T \cup \text{Upper}(O_{5.3})$ and $\Theta_F=\Theta_F \cup \text{Upper}(O_{5.15}) \cup \text{Upper}(O_{6.23})$. According to the relationship in Fig. \ref{fig:pattern},
\begin{eqnarray*}
\Theta_T=&\{&O_{1.0}, O_{2.0}, O_{2.1}, O_{3.0}\sim O_{3.4}, O_{4.0}, O_{4.2}, O_{4.3}, O_{4.6}, O_{4.9}， \\
&& O_{4.11}, O_{5.1},  O_{5.3}, O_{5.16}, O_{6.7}, O_{6.24}, O_{7.20} \}. \\
\Theta_F=&\{& O_{4.10}, O_{5.5}, O_{5.10}, O_{5.15}, O_{6.1}, O_{6.5}, O_{6.8}, O_{6.13}, O_{6.16}, O_{6.20}, O_{6.22}, O_{6.23},\\
&& O_{7.0}, O_{7.1}, O_{7.3}, O_{7.4}, O_{7.6}~O_{7.11},  O_{7.13}, O_{7.14}, O_{7.16}, O_{7.17}, O_{7.19}, O_{7.21},  O_{7.22}, \\
&& O_{7.24}, O_{7.26}, O_{7.27}, O_{7.28}, O_{7.29}, \\
&& O_{8.0}\sim O_{8.9}, O_{8.11}\sim O_{8.24}, O_{9.0}\sim O_{9.11}, O_{9.13}\sim O_{9.16},\\
&& O_{10.0}\sim O_{10.11}, O_{11.0}\sim O_{11.4}, O_{12.0}, O_{12.1,}, O_{13.0}, O_{14.0}\}.
\end{eqnarray*}

\textbf{Step 5:} Now we consider $\varDelta^{G_6^{6}}_{f}$. According to $\Theta_T$ and $\Theta_F$, we can check that $H_6^1 \in O_{2.1} \in \Theta_T$, $H_6^2 \in O_{4.7}$,  $H_6^3 \in O_{4.11} \in \Theta_F$, $H_6^4 \in O_{4.7}$, $H_6^1 \cup H_6^2 \in O_{6.19}$, $H_6^1 \cup H_6^3 \in O_{6.24}$, $H_6^1 \cup H_6^4 \in O_{6.19}$, $H_6^2 \cup H_6^3 \in O_{8.10}$,  $H_6^2 \cup H_6^4 \in O_{8.14} \in \Theta_F$, $H_6^3 \cup H_6^4 \in O_{8.10}$, and, again, for all $10 \leq k$, $O_{k,j} \in \Theta_F$. 
Therefore, in order to make $\chi(\varDelta^{G_6^{7}}_{f})$ be 1, there are three possible cases,
\begin{enumerate}
\item $\varDelta^{G_6^{6}}_{f}=\{\{H_{6}^1\},\{H_{6}^3\},\{H_{6}^1,H_{6}^3\},\emptyset\}$;
\item $\varDelta^{G_6^{6}}_{f}=\{\{H_{6}^1\},\{H_{6}^2\},\{H_{6}^3\},\{H_{6}^4\},\{H_{6}^1,H_{6}^3\},\{H_{6}^1,H_{6}^2,\{H_{6}^1,H_{6}^4\},\emptyset\}$;
\item $\varDelta^{G_6^{6}}_{f}=\{\{H_{6}^1\},\{H_{6}^2\},\{H_{6}^3\},\{H_{6}^4\},\{H_{6}^1,H_{6}^3\},\{H_{6}^2,H_{6}^3\},\{H_{6}^3,H_{6}^4\},\emptyset\}$;

\end{enumerate}
Suppose the first one is true. Then there is one new F-orbits added. Thus, $\Theta_F=\Theta_F \cup \text{Upper}(O_{4.7}) $. According to the relationship in Fig. \ref{fig:pattern},
\begin{eqnarray*}
\Theta_T=&\{&O_{1.0}, O_{2.0}, O_{2.1}, O_{3.0}\sim O_{3.4}, O_{4.0}, O_{4.2}, O_{4.3}, O_{4.6}, O_{4.9}， \\
&& O_{4.11}, O_{5.1},  O_{5.3}, O_{5.16}, O_{6.7}, O_{6.24}, O_{7.20} \}. \\
\Theta_F=&\{& O_{4.7},O_{4.10}, O_{5.0}, O_{5.5}, O_{5.9}, O_{5.10}, O_{5.11}, O_{5.15},\\
&&O_{6.0}\sim O_{6.6}, O_{6.8}, O_{6.9}, O_{6.13},O_{6.15}, O_{6.16}, O_{6.19}\sim  O_{6.23},\\
&& O_{7.0}\sim O_{7.19}, O_{7.21}\sim O_{7.24}, O_{7.26}\sim O_{7.29},  \\
&& O_{8.0}\sim O_{8.24}, O_{9.0}\sim O_{9.16}, \\
&& O_{10.0}\sim O_{10.11}, O_{11.0}\sim O_{11.4}, O_{12.0}, O_{12.1,}, O_{13.0}, O_{14.0}\}.
\end{eqnarray*}
\begin{table}[H]
\centering
\caption{k-combinations of 1-orbits of $G_6^{3}$ \label{table:G_6^3}}
{\begin{tabular}{ | p{1cm}  | p{10cm} |}
\hline 
 k& combinations \\ 
 \hline
$1$& $O_{1.0} \times 2$,  $O_{3.1}$,  $O_{3.2} \times 2$,   $O_{3.4}$ \\ 
\hline
$2$& $O_{2.1}$, $O_{4.0} \times 2$, $O_{4.2} \times 2$,  $O_{4.8} \times 2$,  $O_{4.10} \times 2$,\newline
$O_{6.7} \times 2$, $O_{6.14} \times 2$, $O_{6.23}$, $O_{6.24}$, \\ 
\hline
$3$& $O_{5.3}$, $O_{5.10} \times 2$,  $O_{5.15}$, $O_{7.0} \times 2$, $O_{7.20} \times 2$, $O_{7.25} \times 2$,\newline
 $O_{7.27} \times 4$, $O_{7.29} \times 2$, $O_{9.12} $, $O_{9.13} \times 2$,  $O_{9.14}$. \\
\hline
\end{tabular}}
\end{table}
\textbf{Step 6:} Now we consider $\varDelta^{G_6^{3}}_{f}$. The combinations of the orbits on 1-tuples under $G_6^3$ are shown as Table \ref{table:G_6^3}. Because all the $4$-combinations of $H_3^i$ have at least $8$ elements and for all $i>7$ and $j$, $O_{i,j} \subseteq \Theta_F$, $\varDelta^{G_6^{3}}$ can only have 1-combination, 2-combinations, or 3-combinations. According to $\Theta_T$,  $\varDelta^{G_6^{3}}_{f}$ currently has 6 1-combinations, 8 2-combinations and 3 3-combinations. According to $\Theta_T$ and $\Theta_F$, for the orbits of the 2-combinations and 3-combinations, the free orbits are $O_{4.8}$, $O_{6,14}$ and $O_{7.25}$. Furthermore, $O_{4.8} \in O_{7.25}$. Therefore, there two possible settings of $O_{4.8}$, $O_{6.14}$ and $O_{7.25}$, for $\varDelta^{G_6^{3}}_{f}$ to be one.
One is to set $O_{4.8}$, $O_{6.14}$ and $O_{7.25}$ as F-orbits, and the other one is to set $O_{7.25}$ and $O_{4.8}$ as T-orbits and $O_{6.14}$ as an F-orbit. Suppose the former one is true. We update $\Theta_T$ and $\Theta_F$ accordingly and obtain the followings, 
\begin{eqnarray*}
\Theta_T=&\{&O_{1.0}, O_{2.0}, O_{2.1}, O_{3.0}\sim O_{3.4}, O_{4.0}, O_{4.2}, O_{4.3}, O_{4.6}, O_{4.9}， \\
&& O_{4.11}, O_{5.1},  O_{5.3}, O_{5.16}, O_{6.7}, O_{6.24}, O_{7.20} \}. \\
\Theta_F=&\{& O_{4.7}, O_{4.8}, O_{4.10}, O_{5.0}, O_{5.4}, O_{5.5}, O_{5.9}, O_{5.10}, O_{5.11}, O_{5.13}, O_{5.15},\\
&&O_{6.0}\sim O_{6.6}, O_{6.8}\sim O_{6.10}, O_{6.13}\sim O_{6.16}, O_{6.18}\sim  O_{6.23},\\
&& O_{7.0}\sim O_{7.19}, O_{7.21}\sim O_{7.29},   O_{8.0}\sim O_{8.24}, O_{9.0}\sim O_{9.16}, \\
&& O_{10.0}\sim O_{10.11}, O_{11.0}\sim O_{11.4}, O_{12.0}, O_{12.1,}, O_{13.0}, O_{14.0}\}.
\end{eqnarray*}

\textbf{Step 7:} Now we are ready to consider $\chi(\varDelta^{G_6^{1}}_{f})$ and $\chi(\varDelta_{f_{x_1=1}})$. According to  $\Theta_T$ and $\Theta_F$, currently we have $\chi(\varDelta^{G_6^{1}}_{f})=1$ and $\chi(\varDelta_{f_{x_1=1}})=7$, and the only free orbits are $O_{5.4}, O_{5.6}, O_{5.12}, O_{6.10}, O_{6,12}$ and $O_{6.17}$. The size of these orbits together with their relations are shown in Fig. \ref{fig:001}, where the $a/b$ implies that orbit has totally $a$ elements where $b$ of them contains variable $x_1$. Recall the definition of $\chi(\varDelta)$ in Eq. (\ref{eq:euler}), once adding an orbits $O_{i,j}$ to $\Theta_T$, $\chi(\varDelta^{G_6^{1}}_{f})$ and $\chi(\varDelta_{f_{x_1=1}})$ increase 
$(-1)^{i+1}|O_{i,j}|$ and $(-1)^{i}|O_{i,j}|$, respectively. Note that whatever the types these free orbits own, the values of $\chi(\varDelta^{G_6^i}_{f})$, $2 \leq i \leq 11$ remain unchanged. Therefore, for this subcase, it suffices to show there is no setting of the types of these free orbits can satisfy both $\chi(\varDelta^{G_6^1}_{f})=1$  and $\chi(\varDelta_{f_{x_1=1}})=1$. In order to make $\chi(\varDelta^{G_6^1}_{f})$ be 1, there are two possible cases, 
\begin{enumerate}
\item T-orbits: $O_{6.12}, O_{6.17}, O_{5.6}, O_{5.12}$; F-orbits: $O_{6.10}, O_{5.4}$.
\item T-orbits: $O_{6.12}, O_{6.17}, O_{5.6}, O_{5.12}, O_{6.10}, O_{5.4}$.
\end{enumerate} 
One can check that $\chi(\varDelta_{f_{x_1=1}}) \neq 1$ in neither of the above cases. Then the checking process will back to the previous step and consider other possible cases.
\begin{figure}[H]
\begin{center}
\includegraphics[width=3in]{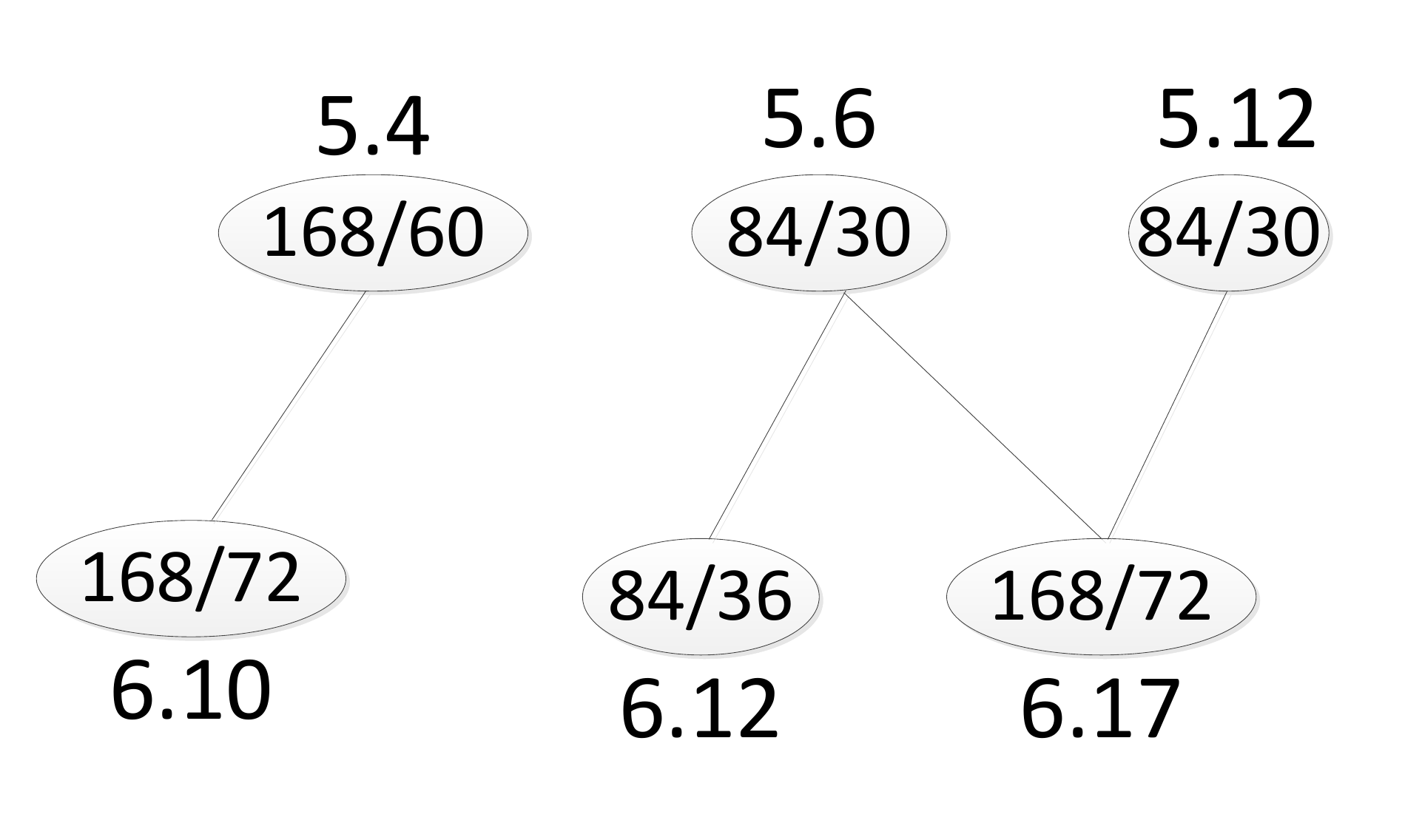} 
\end{center} 
\caption{\textbf{Remained free orbits}}
\label{fig:001}
\end{figure}

\section*{References}

\bibliography{sigproc}

\end{document}